\documentclass[12pt, reqno]{amsart}

\usepackage{amsmath, amssymb, amsfonts, amsbsy, amsthm, latexsym}
\usepackage{graphicx}
\usepackage{verbatim}
\usepackage{color}
\usepackage{xy}
\usepackage{pstricks, pst-node}

\newtheorem{theorem}{Theorem}[section]
\newtheorem{proposition}[theorem]{Proposition}
\newtheorem{lemma}[theorem]{Lemma}
\newtheorem{corollary}[theorem]{Corollary}

\theoremstyle{definition}

\newtheorem{example}[theorem]{Example}

\numberwithin{equation}{section}

\numberwithin{equation}{section}

\usepackage{enumerate}

\def\EE{{\mathbb{E}}}
\def\0{{\bar{0}}}

\def\RR{{\mathbb{R}}}

\def\R{{\mathbb{R}}}

\def\RR{{\mathbb{R}}}

\begin{document}
\title[Error bounds for consistent reconstruction]{Error bounds for consistent reconstruction: \\ random polytopes and coverage processes
}
%%%\title[Consistent reconstruction, random polytopes, coverage processes]{Error bounds for consistent reconstruction: \\ random polytopes and coverage processes}

\author{Alexander M. Powell}
\address{Vanderbilt University, Department of Mathematics,
Nashville, TN 37240, USA}
\email{alexander.m.powell@vanderbilt.edu}

\author{J. Tyler Whitehouse}
%\address{Vanderbilt University, Department of Mathematics, Nashville, TN 37240, USA}
%\email{jonathan.t.whitehouse@vanderbilt.edu}
\email{tyler.whitehouse@gmail.com}

%General info
%\subjclass{Primary ?????, ?????; Secondary ?????}

\date{April 9, 2013}


\begin{abstract}
Consistent reconstruction is a method for producing an estimate $\widetilde{x} \in \R^d$ of a signal $x\in \R^d$ if one is given a
collection of $N$ noisy linear measurements $q_n = \langle x, \varphi_n \rangle + \epsilon_n$, $1 \leq n \leq N$, that have been corrupted by i.i.d.
uniform noise $\{\epsilon_n\}_{n=1}^N$.  
We prove mean squared error bounds for consistent reconstruction when the measurement vectors $\{\varphi_n\}_{n=1}^N\subset \R^d$ are drawn
independently at random from a suitable distribution on the unit-sphere $\mathbb{S}^{d-1}$.
Our main results prove that the mean squared error (MSE) for consistent reconstruction is of the optimal order
$\mathbb{E}\|x - \widetilde{x}\|^2 \leq K\delta^2/N^2$ under general conditions on the measurement vectors.
We also prove refined MSE bounds when the measurement vectors are i.i.d. uniformly distributed on the unit-sphere $\mathbb{S}^{d-1}$
and, in particular, show that in this case the constant $K$ is dominated by $d^3$, the cube of the ambient dimension.
The proofs involve an analysis of random polytopes using coverage processes on the sphere.
\end{abstract}

\maketitle

\section{Introduction}

We consider the problem of estimating an unknown signal $x\in \R^d$ from a collection of $N \geq d$ noisy linear measurements
\begin{equation} \label{qn-def}
q_n = \langle x, \varphi_n \rangle + \epsilon_n, \ \ \  1 \leq n \leq N,
\end{equation}
where $\{\varphi_n\}_{n=1}^N \subset \R^d$ is a known spanning set for $\R^d$, and where the unknown noise $\{\epsilon_n\}_{n=1}^N$
has been independently drawn according to the uniform distribution on a known interval $[-\delta, \delta]$.
{\em Consistent reconstruction} is a method for producing an estimate $\widetilde{x} \in \R^d$ of $x$ from the noisy measurements \eqref{qn-def}.
Consistent reconstruction selects $\widetilde{x}$ as any solution to the linear feasibility problem
\begin{equation} \label{cr-def}
\forall \ 1 \leq n \leq N, \ \ \ \ | \langle \widetilde{x}, \varphi_n \rangle - q_n | \leq \delta.
\end{equation}
In other words, \eqref{cr-def} simply seeks an estimate $\widetilde{x}$ that is consistent with the knowledge that the noise is bounded in $[-\delta,\delta]$.
Our main contribution in this work is to provide sharp bounds on the mean squared estimation error associated with consistent reconstruction and to quantify how accurately \eqref{cr-def} recovers $x$ from the measurements  \eqref{qn-def} as a function of the number of measurements $N$ and the dimension $d$.  
In our analysis, $\{\varphi_n\}_{n=1}^N\subset \R^d$ will be i.i.d. random vectors drawn from a suitable distribution on the unit-sphere $\mathbb{S}^{d-1}$, and special attention will be given to the case when each $\varphi_n$ is uniformly distributed on $\mathbb{S}^{d-1}$.

Consistent reconstruction has received particular attention in the signal processing literature as a method for
recovering signals from quantized samples.   
Deterministic round-off errors that arise in quantization are frequently modeled using uniform noise.  For example, the use of uniform noise models in quantization is typically justified by dithering, \cite{RG}, or with high resolution asymptotics as the quantizer step-size approaches zero, \cite{JWW}.   
Consistent reconstruction and its variants have been shown to be an effective method for signal recovery in memoryless scalar quantization (MSQ), \cite{GVT, VT94b, RG, Z01, Z03}, Sigma-Delta ($\Sigma\Delta$) quantization, \cite{VT94a}, compressed sensing, \cite{JHF}, and finite rate of innovation sampling, \cite{JB06}.  
A key point often observed in practice is that when compared with linear reconstruction, consistent reconstruction can reduce the mean squared reconstruction error by an extra multiplicative factor that scales inversely with the sampling rate.

The structure of uniformly distributed noise plays an essential role in our analysis of consistent reconstruction \eqref{cr-def}.  
It is useful to note that estimation with uniform noise in \eqref{qn-def} falls outside of several classical approaches to estimation theory.
For example, asymptotic normality theorems in maximum likelihood estimation (MLE) typically
require suitable smoothness assumptions on the underlying noise distribution which uniform noise does not satisfy, e.g., \cite{F}.  
Similarly, the multiparameter Cramer-Rao bound gives lower bounds for minimum variance unbiased estimation, but also requires suitable regularity on the 
noise distribution.
Finally, recall that linear estimation commonly yields mean squared error (MSE) bounds of order $1/N$,
and this is, for example, optimal for Gaussian noise with respect to the Cramer-Rao bound.  However,  when dealing with uniform noise, linear estimation is
typically sub-optimal and it is possible to provide more accurate recovery than MSE of order $1/N$.
 We provide a rigorous analysis of the mean squared error in consistent reconstruction for the general estimation problem \eqref{qn-def}
with suitable random measurement vectors $\{\varphi_n\}_{n=1}^N \subset \mathbb{R}^d$, and we prove that the
 mean squared error is of the optimal order $K/N^2$ with precise control on the constant $K$.  A sample consequence of this is that $N = \mathcal{O}(d^{3/2})$ random measurements will suffice to achieve highly accurate mean squared reconstruction error (compared to $\mathcal{O}(d^2)$ measurements with linear reconstruction).

Consistent reconstruction \eqref{cr-def} has a simple geometrical interpretation.  
Since the bounded noise satisfies $|\epsilon_n| \leq \delta$, each noisy measurement $q_n$ gives the information that the unknown true signal $x\in\R^d$ lies
in the $2\delta$-thick slab
\begin{equation}\label{Sn-def} S_n = S_n(q_n, \varphi_n)= \{ u \in \R^d : | \langle u, \varphi_n \rangle - q_n | \leq \delta \}.\end{equation}
Consequently, the consistency equations \eqref{cr-def} are equivalent to requiring that $\widetilde{x}\in\R^d$ lies in the {\em consistent reconstruction polytope} defined by
\begin{equation} \label{QN-def} Q_N =  \bigcap_{n=1}^N S_n.\end{equation}
The assumption that $\{\varphi_n\}_{n=1}^N$ spans $\R^d$ ensures that $Q_N$ is a compact set.  Moreover, $Q_N$ is almost surely a nondegenerate polytope with nonempty interior.  Since the true signal $x$ is always contained in the polytope $Q_N$, it is clear that the system \eqref{cr-def} is feasible.

The main object of interest in this article will be the worst case error associated with consistent reconstruction.
Recalling that \eqref{cr-def} generally has infinitely many solutions, the worst case error may be defined as follows.
If $\widetilde{x} \in Q_N$ is any solution to the consistent reconstruction system \eqref{cr-def}, then
the error $(x - \widetilde{x})$ lies in the following {\em error polytope} 
\begin{equation} \label{PN-def}
P_N =  \bigcap_{n=1}^N E_n
\end{equation}
where
\begin{equation} \label{En-def}
E_n = \{ u \in \R^d : | \langle u, \varphi_n \rangle - \epsilon_n | \leq \delta \}.
\end{equation}
The polytope $P_N$ is obtained by translating $Q_N$ to the origin by the vector $x$.
In terms of the polytopes $P_N$ and $Q_N$, the worst case error associated with consistent reconstruction can be defined as
\begin{equation} \label{WN-def}
W_N = \max \{ \| u \| :  u \in P_N \} = \max \{ \| u - x \| : u \in Q_N \}.
\end{equation}
Viewed geometrically, the worst case error $W_N$ is precisely the radius of the smallest closed ball centered at 0 that contains the error polytope $P_N$.

\subsection*{Overview and main results}
The main contribution of this article is to prove that the expected worst case error squared for consistent reconstruction is of the optimal order
\begin{equation} \label{overview-eq}
\mathbb{E}|W_N|^2 \leq \frac{K\delta^2}{N^2},
\end{equation}
where the constant $K>0$ depends on the distribution of the random vectors $\{\varphi_n\}_{n=1}^N$.
Our first main result, Theorem \ref{main-thm1}, proves the mean squared error bound \eqref{overview-eq} under general assumptions on the i.i.d. unit-norm random vectors 
$\{\varphi_n\}_{n=1}^N \subset \mathbb{S}^{d-1}$.  Our second main result, Theorem \ref{mainthm2}, proves a refined version of \eqref{overview-eq} 
when $\{\varphi_n\}_{n=1}^N \subset \mathbb{S}^{d-1}$ are i.i.d. random vectors that are uniformly distributed on $\mathbb{S}^{d-1}$ and, in particular, shows that
the constant $K$ in \eqref{overview-eq} is dominated by $d^3$, the cube of the ambient dimension.

The paper is organized as follows.  Section \ref{background:sec} provides background on estimation with uniform noise, and also discusses preliminaries and notation
such as surface measure on the sphere and geodesic $\epsilon$-nets.  Section \ref{unidirection-sec} analyzes the size of the error polytope $P_N$ in a fixed direction.
Section \ref{cover-sec} provides necessary background and results on coverage processes on the sphere which will be used in the proofs of our main theorems.
Section \ref{genthm-sec} states and proves our first main theorem, Theorem \ref{main-thm1}, which shows that consistent reconstruction
satisfies \eqref{overview-eq} under general assumptions.  
Section \ref{unif-sec} states and proves our second main theorem, Theorem \ref{mainthm2}, which shows that for random vectors that are uniformly 
distributed on the unit-sphere consistent reconstruction satisfies \eqref{overview-eq} with a constant $K$ that is dominated by $d^3$.

\section{Background and notation}
\label{background:sec}
\subsection{Estimation with uniform noise and consistent reconstruction} \label{est-background:sec}

In this section we briefly recall some background on estimation with uniform noise and consistent reconstruction.  This will help provide perspective on our main
results.

Begin by recalling linear reconstruction. Let $N \geq d$.  Suppose that $\{\varphi_n\}_{n=1}^N \subset \R^d$ spans $\R^d$ and that $\{f_n\}_{n=1}^N \subset \R^d$ is
any dual frame satisfying
\begin{equation}
\forall x \in \R^d, \ \ \  x = \sum_{n=1}^N \langle x, \varphi_n \rangle f_n.
\end{equation}
Suppose for the moment that $\{\epsilon_n\}_{n=1}^N$ are simply independent zero mean random variables with variance $\mathbb{E}|\epsilon_n|^2 = \sigma^2$,
and that one wishes to estimate $x\in \R^d$ from the noisy measurements $q_n = \langle x, \varphi_n \rangle  + \epsilon_n$, $1 \leq n \leq N$.
Using the dual frame $\{f_n\}_{n=1}^N$ to linearly reconstruct $\widetilde{x}\in\R^d$ by
\begin{equation}
\widetilde{x} = \sum_{n=1}^N q_n f_n
\end{equation}
yields the estimation error
\begin{equation} \label{linrec-mse}
\forall x \in \R^d, \ \ \  \mathbb{E}\| x- \widetilde{x}\|^2 =  \sigma^2 \sum_{n=1}^N \| f_n\|^2.
\end{equation}
If each $\varphi_n$ is assumed to be unit-norm, $\|\varphi_n\|=1$, then the MSE in \eqref{linrec-mse} is bounded below by
\begin{equation} \label{utf-mse}
\forall x \in \R^d, \ \ \ \mathbb{E}\|x-\widetilde{x}\|^2 \geq \frac{d^2 \sigma^2}{N},
\end{equation}
for example, see \cite{GKK}.  Moreover, if each $\| \varphi_n\| =1$ then equality holds in \eqref{utf-mse} precisely when $\{\varphi_n\}_{n=1}^N$
is a unit-norm tight frame for $\R^d$ and when $f_n = \frac{d}{N} \varphi_n$ is taken as the associated canonical dual frame, e.g., see \cite{GKK}.  
The case of unit-norm tight frames yields the mean squared error $\mathbb{E}\| x - \widetilde{x}\|^2 = {d^2 \sigma^2}/{N}$.

It is natural to ask how much one can improve on the standard $1/N$ mean squared accuracy \eqref{utf-mse} if one uses nonlinear reconstruction.
We are specifically interested in the case of uniform noise, and assume henceforth that $\{ \epsilon_n\}_{n=1}^N$ are i.i.d. uniform random variables on $[-\delta, \delta]$.

We begin by mentioning a Bayesian lower bound due to Rangan and Goyal, \cite{RG}.
Suppose that the vectors $\{ \varphi_n \}_{n=1}^{\infty} \subset \R^d$ are unit-norm and that $x\in \R^d$ is an absolutely continuous random vector and that 
$q_n = \langle x, \varphi_n \rangle + \epsilon_n$. 
Let $\widetilde{x}_N = \widetilde{x}_N( \{q_n\}_{n=1}^N, \{ \varphi_n \}_{n=1}^N)$ be any estimator which maps
each observed input $(\{q_n\}_{n=1}^N, \{ \varphi_n \}_{n=1}^N)$ to an estimate $\widetilde{x}_N\in\R^d$ of the signal $x\in \R^d$.  
It was shown in \cite{RG} that the mean squared error is bounded below in the following manner:
\begin{equation} \label{rg-lower-bnd}
{\rm lim \thinspace inf}_{N \to \infty} \ N^2 \thinspace \mathbb{E} \| x - \widetilde{x}_N\|^2 >0.
\end{equation}
Unlike \eqref{utf-mse} the expectation in \eqref{rg-lower-bnd} is taken over both a random signal $x$ and the noise $\{\epsilon_n\}_{n=1}^N$.
The lower bound \eqref{rg-lower-bnd} 
shows that for estimation in the setting of uniform noise one cannot expect MSE that is more accurate than $1/N^2$.  Related lower bounds for
quantization problems can be found in \cite{GVT}.

In \cite{RG}, Rangan and Goyal proposed an estimation algorithm for \eqref{qn-def} that achieves the optimal $1/N^2$ error rate.
Their algorithm starts with an arbitrary $\widetilde{x}_0 \in \R^d$, and
iteratively produces estimates $\widetilde{x}_n \in \R^d$ with the following soft-thresholding algorithm
\begin{equation} \label{rg-alg}
\widetilde{x}_n = \widetilde{x}_{n-1} + \frac{\varphi_n \thinspace T_{\delta} (q_n - \langle \widetilde{x}_{n-1}, \varphi_n \rangle)}{\|\varphi_n\|^2},
\end{equation}
where the soft-thresholding function $T_\delta:\mathbb{R} \to \mathbb{R}$ is defined by
\begin{equation} \label{soft-thresh}
T_{\delta}(t) = 
\begin{cases}
t - \delta, & \hbox{ if } t > \delta,\\
 0, & \hbox{ if }  |t| \leq  \delta,\\
t + \delta, & \hbox{ if } t < -\delta. 
\end{cases}
\end{equation}
The error analysis of the algorithm \eqref{rg-alg} in \cite{RG} assumed that $\{ \varphi_n\}_{n=1}^N$ are independent identically distributed versions 
of a random vector $\varphi$ satisfying the following condition
\begin{equation} \label{rg-exp-cond}
\exists \alpha>0, \ \forall x \in \mathbb{S}^{d-1} , \ \ \ \mathbb{E}|\langle x, \varphi \rangle| \geq \alpha>0.
\end{equation}
It was proven in \cite{RG} that if the i.i.d. random vectors $\{ \varphi_n \}_{n=1}^N$ satisfy \eqref{rg-exp-cond} then
\begin{equation}
\forall s<1, \ \forall x \in \R^d, \ \ \ \lim_{N\to \infty} N^{2s} \| x - \widetilde{x}_N\|^2 =0, \ \ \ \hbox{ almost surely.}
\end{equation}
Moreover, the mean squared error was later proven to satisfy $\mathbb{E}\|x- \widetilde{x}_N\|^2 \leq C/N^2$ for a suitable constant $C>0$ in \cite{AP-RG}.
The algorithm \eqref{rg-alg} need not produce globally consistent estimates but instead employs local updates that can be sensitive to ordering issues.

Consistent reconstruction \eqref{cr-def} provides a maximum likelihood estimate (MLE) for \eqref{qn-def}.  
Let $q=q(x)$ be the $N \times 1$ random vector that is parametrized by $x$ and whose $n$th entry is $q_n  = \langle x, \varphi_n\rangle + \epsilon_n$.
The associated likelihood function is given by $\mathcal{L}(x | q) = \chi_{Q_N}(x)$,
where $\chi_{Q_N}$ is the indicator function of the consistent reconstruction polytope $Q_N$ in  \eqref{QN-def}.  
Thus, the likelihood function $\mathcal{L}(x|q)$
is maximized precisely when $x$ is a consistent estimate satisfying \eqref{cr-def}.
Moreover, asymptotic normality results from MLE do not apply here since $\mathcal{L}(x|q)$ does not satisfy the smoothness assumptions that
are typically needed, \cite{F}.

We conclude this section with the following simple one-dimensional example to provide intuition into the desired $1/N^2$ error rate for consistent reconstruction. 
\begin{example}[Consistent reconstruction in one dimension] \label{cr-1dim-ex}
Let $x \in \R$ and $q_n = x + \epsilon_n$, where $\{\epsilon_n \}_{n=1}^N$ are i.i.d. uniform random variables on $[-\delta,\delta]$.
Consider the problem of estimating $x$ from the noisy observations $\{q_n\}_{n=1}^N$.
In this one-dimensional example, consistent reconstruction simply selects any estimate $\widetilde{x}_N$ that lies in the interval $I_N=[A_N,B_N]$
where
$$A_N  = \max \{ q_n - \delta : 1 \leq n \leq N\} \ \ \ \ \ \hbox{ and } \ \ \ \ \ B_N = \min \{ q_n + \delta: 1 \leq n \leq N\},$$
and the associated worst case estimation error is given by
$$w_N  = \max \{ |x - A_N|, |x - B_N| \}.$$
Elementary order statistics computations show that
$$\mathbb{E} |x-A_N|^2 =  \mathbb{E} |x-B_N|^2 =\frac{8\delta^2}{(N+1)(N+2)}$$
and that the worst error $w_N$ satisfies the following mean squared error bound
$$\mathbb{E} |w_N|^2 = \frac{14\delta^2}{(N+1)(N+2)}.$$
An important technical issue for analyzing \eqref{cr-def} in $\mathbb{R}^d$ will be that the geometry of the error polytope \eqref{PN-def} becomes non-trivial in higher dimensions.
\end{example}

\subsection{Preliminaries and notation}  In this section we collect some necessary notation and background results concerning measure on the sphere and epsilon-nets.

We shall denote an open spherical cap on the unit-sphere $\mathbb{S}^{d-1}\subset \R^d$ with center $\varphi \in \mathbb{S}^{d-1}$ and angular radius 
$0<\theta<\pi$ by
\begin{equation}
{\rm Cap}(\varphi, \theta) = \{ u \in \mathbb{S}^{d-1} : \langle u, \varphi \rangle > \cos \theta \}.
\end{equation}
If $0<\theta<\pi/2$ then the relative measure (normalized with respect to $\mathbb{S}^{d-1}$) of ${\rm Cap}(\varphi, \theta)$
is given by, e.g., \cite{BCL10}, 
\begin{equation} \label{rel-cap-meas}
r_{d-1}(\theta) = \frac{\Gamma(\frac{d}{2})}{\sqrt{\pi}\Gamma(\frac{d-1}{2})} \int_0^{\theta} (\sin u)^{d-2} d u.
\end{equation}
It is useful to note, e.g., \cite{W46}, that when $d\geq 2$ the
 constant 
 \begin{equation} \label{CdDef}
 C_d = \frac{\Gamma(\frac{d}{2})}{\sqrt{\pi}\Gamma(\frac{d-1}{2})}
 \end{equation}
 satisfies
 \begin{equation}\label{C-d-estimate}
\sqrt{1 - \frac{1}{d} \thinspace} \thinspace \sqrt{\frac{d-1}{2\pi}} \leq C_d \leq \sqrt{\frac{d-1}{2\pi}}
\ \ \ \hbox{ and } \ \ \ \lim_{d\to\infty}\frac{C_d}{\sqrt{d}}=\frac{1}{\sqrt{2\pi}}.
\end{equation}

Fix any $x_0 \in \mathbb{S}^{d-1}$ and let the random vector $\varphi\in \mathbb{R}^{d-1}$ be uniformly distributed on the unit-sphere $\mathbb{S}^{d-1}$.
Note that, by rotation invariance, the distribution of the random variable $Z= |\langle x_0, \varphi \rangle|$ does not depend on $x_0$.
The pdf of the random variable $Z$ is given by, e.g., \cite{Stam82},
\begin{equation} \label{Zpdf}
f_Z(z) = 
\begin{cases}
2C_d (1-z^2)^{\frac{d-3}{2}}, & \hbox{ if } z \in [0,1],\\
0, & \hbox{ if } z \not\in [0,1].
\end{cases}
\end{equation}

The geodesic distance between between two points $x,y \in \mathbb{S}^{d-1}$ on the sphere will be denoted by
$$d(x,y) = \arccos( \langle x, y \rangle).$$
Given $\epsilon>0$, we say that a set $\mathcal{N}_{\epsilon} \subset \mathbb{S}^{d-1}$ is a geodesic $\epsilon$-net of $\mathbb{S}^{d-1}$ if 
$$\forall x \in \mathbb{S}^{d-1}, \ \exists z \in \mathcal{N}_{\epsilon}, \ \ \hbox{ such that } \ \ d(x,z) \leq \epsilon.$$
A standard argument shows that if $d\geq 2$ then there exist geodesic $\epsilon$-nets $\mathcal{N}_{\epsilon}$ of $\mathbb{S}^{d-1}$ with cardinality satisfying 
\begin{align} \label{eps-net-card-eq}
\# \left( \mathcal{N}_{\epsilon}  \right) &\leq \frac{1}{r_{d-1}(\epsilon/2)}
 \leq \frac{\sqrt{\pi} \thinspace \Gamma(\frac{d-1}{2})}{\Gamma(\frac{d}{2})} \left( \frac{\pi}{2}\right)^{d-2} \frac{(d-1)}{(\epsilon/2)^{d-1}}
 \leq \left( \frac{8}{\epsilon}\right)^{d-1}.
\end{align}
For example, if $\mathcal{N}_{\epsilon}$ is a maximal $\epsilon$-separated (with respect to geodesic distance) subset of $\mathbb{S}^{d-1}$ then $\mathcal{N}_{\epsilon}$ is a geodesic
$\epsilon$-net for $\mathbb{S}^{d-1}$ and satisfies \eqref{eps-net-card-eq}.

\section{Error polytope size in a fixed direction} \label{unidirection-sec}

In this section we study the radial size of the error polytope $P_N$ in a {\em fixed} direction.  
Given a unit-vector $\psi \in \mathbb{S}^{d-1}$, define the radial size of the error polytope in the direction $\psi$ by
\begin{equation}
R_N(\psi) = \max \{ r \geq 0 : r\psi \in P_N \}.
\end{equation}
Note that $W_N \geq R_N$ since the worst case error \eqref{WN-def} satisfies $$W_N = \sup \{R_N(\psi): \psi \in \mathbb{S}^{d-1} \}.$$

The next result follows immediately from Example \ref{cr-1dim-ex} and provides a simple lower bound on $\mathbb{E}|R_N(\psi)|^2$
for a general choice of $\{\varphi_n\}_{n=1}^N\subset \mathbb{S}^{d-1}$.

\begin{proposition} \label{GenRbndProp}
Let $\{\varphi_n \}_{n=1}^N \subset \mathbb{S}^{d-1}$ be arbitrary and let $\psi \in \mathbb{S}^{d-1}$.
Then the worst case error in the direction $\psi$ for consistent reconstruction satisfies
$$\forall N, \ \ \ \mathbb{E}|R_N(\psi)|^2 \geq \frac{8\delta^2}{(N+1)(N+2)}.$$
\end{proposition}

\begin{proof}
Note that
$$R_N(\psi) = \min \{ x_n : 1 \leq n \leq N \},$$
where
$$x_n = 
\begin{cases}
{(\epsilon_n+\delta)}/{|\langle \varphi_n, \psi \rangle|  }, & \hbox{ if } \langle \varphi_n, \psi \rangle \geq 0, \\
{(\delta - \epsilon_n)}/{|\langle \varphi_n, \psi \rangle|}, & \hbox{ if } \langle \varphi_n, \psi \rangle < 0. \\
\end{cases}
$$
Since $|\langle \varphi_n, \psi \rangle| \leq 1$ and since $(\epsilon_n+\delta)$ and $(\delta -\epsilon_n)$ are both uniformly distributed on $[0, 2\delta]$ it follows that
$x_n \geq \xi_n$, where $\xi_n$ is uniformly distributed on $[0,2\delta]$.  The proof now follows from Example \ref{cr-1dim-ex}.
\end{proof}

In the remainder of this section, we study $R_N(\psi)$ in the case when the vectors $\{ \varphi_n \}_{n=1}^N \subset \R^d$ used to define the error polytope $P_N$ are i.i.d.
uniformly distributed random vectors on the unit-sphere $\mathbb{S}^{d-1}$.  In this case, rotation invariance implies
that the distribution of $R_N=R_N(\psi)$ is independent of $\psi$.
We will make use of the following lemma whose proof follows from similar steps as in Proposition \ref{GenRbndProp}.

\begin{lemma} \label{YN-lemma}
Let $\xi$ be a uniform random variable on $[0, 2\delta]$, and define the random variable $Z=|\langle e_0, \varphi \rangle|$, where
$e_0 = (1, 0, 0, \cdots, 0) \in \R^d$ and the random vector $\varphi$ is uniformly distributed on unit-sphere $\mathbb{S}^{d-1}$.
Let $\xi$ and $Z$ be independent and define the random variable $X= {\xi}/{Z}.$
Let $\{X_n\}_{n=1}^N$ be $N$ independent versions of the random variable $X$ and consider the associated order statistic
$$Y_N = \min \{ X_n : 1 \leq n \leq N \}$$
The random variable $R_N$ has the same distribution as the random variable $Y_N$.
\end{lemma}

\begin{theorem} \label{RadialThmDimd}
Let $d \geq 2$.  Suppose that $\{\varphi_n\}_{n=1}^N\subset \mathbb{S}^{d-1}$ are i.i.d. uniformly distributed random vectors on $\mathbb{S}^{d-1}$.
Then for $N \geq 3$ the worst case error in the direction $\psi \in \mathbb{S}^{d-1}$ for consistent reconstruction satisfies
\begin{equation} \label{RadThmDimdBnd}
\mathbb{E}|R_N(\psi)|^2 = \frac{2\delta^2(d-1)^2}{(C_d)^2(N+1)(N+2)} + 2\delta^2\alpha_{d,N},
\end{equation}
where $C_d$ is as in \eqref{CdDef} and
\begin{equation}
 - \left( \frac{2C_d}{(d-1)} \right) \left( 1 - \frac{C_d}{(d-1)} \right)^{N+1}
 \leq  \alpha_{d,N} \leq  54 (C_d)^2 \left( 1 - \frac{2 C_d}{(d-1)} \right)^N.
\end{equation}
Note that by \eqref{C-d-estimate}, $0< (1 - \frac{C_d}{d-1})<1$ and $0 < (1 - \frac{2C_d}{d-1}) <1$.
\end{theorem}

\begin{proof}
By Lemma \ref{YN-lemma} we need to estimate the following integral
\begin{equation} \label{RadialThmDimdEq1}
\mathbb{E}|R_N(\psi)|^2 = \mathbb{E}|Y_N|^2 = 2 \int_0^{\infty} \lambda \Pr[Y_N >\lambda] d\lambda =
2 \int_0^{\infty} \lambda \left( \Pr[X >\lambda] \right)^N d\lambda.
\end{equation}

\noindent {\em Step I.}  If $0 \leq \lambda \leq 2\delta$ then using \eqref{Zpdf} gives
\begin{align}
\Pr [ X > \lambda] &= \Pr [ \xi > \lambda Z] = \int_0^1\Pr[\xi >\lambda z] f_Z(z) dz \notag \\
& = \int_0^1 \left( \frac{2\delta - \lambda z}{2\delta} \right) 2 C_d(1-z^2)^{\frac{d-3}{2}} dz \notag \\
&= \left( 1 - \frac{\lambda C_d}{\delta(d-1)} \right).  \label{stepI-eq1}
\end{align}
A computation shows that
\begin{align}
\int_0^{2\delta} & \lambda \left( \Pr[X >\lambda] \right)^N d\lambda = \int_0^{2\delta} \lambda  \left( 1 - \frac{\lambda C_d}{\delta(d-1)} \right)^N d\lambda \notag \\
&= \frac{\delta^2(d-1)^2}{(C_d)^2} \left( \frac{1}{(N+1)(N+2)} - \left( 1 - \frac{2C_d}{(d-1)} \right)^{N+1} \left( \frac{1}{N+1} - \frac{(1 - \frac{2C_d}{(d-1)})}{N+2} \right) \right).
\label{step1-eq2}
\end{align}

\noindent {\em Step II.} If $\lambda \geq 2\delta$ then then \eqref{stepI-eq1} gives
$$\Pr[X>\lambda] \leq \Pr[X>2\delta] \leq  \left( 1 - \frac{2 C_d}{(d-1)} \right).$$
Thus
\begin{align}
\int_{2\delta}^{6\pi\delta C_d} \lambda  \left( \Pr[X >\lambda] \right)^N d\lambda &\leq 
\int_{2\delta}^{6\pi\delta C_d} \lambda \left( 1 - \frac{2 C_d}{(d-1)} \right)^N d\lambda \notag \\
& \leq 18\pi^2\delta^2 (C_d)^2 \left( 1 - \frac{2 C_d}{(d-1)} \right)^N. \label{stepII-eq1}
\end{align}

\noindent {\em Step III.} If $\lambda \geq 2 \delta$ then a computation using \eqref{Zpdf} and $d\geq2$ shows that
\begin{align}
\Pr[X>\lambda] &= 2 C_d \int_0^{2\delta/\lambda} \left( \frac{2\delta - \lambda z}{2\delta}\right) (1-z^2)^{\frac{d-3}{2}} dz \notag \\
& \leq 2 C_d \int_0^{2\delta/\lambda}  (1-z^2)^{-\frac{1}{2}} dz \notag \\
& \leq\frac{2\pi \delta C_d}{\lambda}. \label{stepIII-eq1}
\end{align}
If $d\geq 2$ then \eqref{C-d-estimate} along with $\Gamma(1/2)=1/\sqrt{\pi}$ and $\Gamma(3/2) = \sqrt{\pi}/2$ implies that $\frac{1}{3} \leq (1 - \frac{2C_d}{d-1})$.
Equation \eqref{stepIII-eq1} implies that
\begin{align}
\int_{6\pi\delta C_d}^{\infty} \lambda \left( \Pr[X>\lambda] \right)^N d \lambda
& \leq \int_{6\pi\delta C_d}^{\infty} \lambda \left(  \frac{2 \pi \delta C_d}{\lambda} \right)^N d\lambda \notag \\
& = \left( \frac{1}{3} \right)^N \frac{36\pi^2 \delta^2 (C_d)^2}{(N-2)}\notag \\
& \leq 36 \pi^2\delta^2 (C_d)^2 \left( 1 - \frac{2 C_d}{(d-1)} \right)^N. \label{stepIII-eq2}
\end{align}
Combining \eqref{step1-eq2}, \eqref{stepII-eq1}, \eqref{stepIII-eq2} and \eqref{RadialThmDimdEq1} now yields the desired conclusion \eqref{RadThmDimdBnd}.
\end{proof}

\begin{corollary} \label{lower-bnd-cor}
Suppose that $\{\varphi_n\}_{n=1}^N\subset \mathbb{S}^{d-1}$ are i.i.d. uniformly distributed random vectors on $\mathbb{S}^{d-1}$.
Then
$${\rm lim \thinspace inf}_{N \to \infty} N^2 \thinspace \mathbb{E}|W_N|^2 \geq \lim_{N\to \infty} N^2 \thinspace \EE |R_N|^2 = 2 \delta^2\left(\frac{d-1}{2\,C_d}\right)^2
\geq \pi \delta^2 (d-1).
$$
\end{corollary}

\section{Consistent reconstruction and coverage processes} \label{cover-sec}

In this section we describe a useful connection between the worst case error $W_N$ and coverage processes on the sphere $\mathbb{S}^{d-1}$.  
This relationship will play a central
role in the proofs of our upper bounds on $\mathbb{E}|W_N|^2$ that appear in the subsequent sections.

\subsection{Worst case error and coverage processes}
The expected worst case error squared $\mathbb{E}|W_N|^2$ can be represented as
\begin{equation} \label{MSE-intbyparts}
\mathbb{E}|W_N|^2 = 2 \int_0^{\infty} \lambda \ \Pr [ W_N > \lambda ] \thinspace d \lambda.
\end{equation}
So, a main step towards bounding $\mathbb{E}|W_N|^2$ is to bound the probability $\Pr [ W_N > \lambda ]$. 
The next lemma shows that bounding $\Pr [ W_N > \lambda ]$ can be reformulated as a coverage problem.

\begin{lemma} \label{cover-lem}
For each $\lambda>0$ define the set
\begin{equation} \label{Bn-def}
B_n=B_n(\lambda) = \left\{ u \in \mathbb{S}^{d-1} : \langle u , \varphi_n \rangle > \frac{\epsilon_n+\delta}{\lambda} \  
\hbox{ or } \  \langle u , \varphi_n \rangle < \frac{\epsilon_n - \delta}{\lambda}  \right\}.
\end{equation}
Then
\begin{equation} \label{coverage-eq2}
\forall \thinspace \lambda>0, \ \ \ \Pr [ W_N \geq \lambda ] = \Pr \left( \mathbb{S}^{d-1} \not\subset \bigcup_{n=1}^N B_n(\lambda) \right).
\end{equation}
In particular,
\begin{equation} \label{coverage-eq}
\forall \thinspace \lambda>0, \ \ \ \Pr [ W_N > \lambda ] \leq \Pr \left( \mathbb{S}^{d-1} \not\subset \bigcup_{n=1}^N B_n(\lambda) \right).
\end{equation}
\end{lemma}

\begin{proof}
Let $E_n$ be as in \eqref{En-def} and define the set
\begin{align*}
A_n =A_n(\lambda)= ({\lambda \mathbb{S}^{d-1}) \backslash E_n} 
= \{ u \in \lambda \mathbb{S}^{d-1} : \langle u , \varphi_n \rangle >  \epsilon_n + \delta \  \hbox{ or } \  \langle u , \varphi_n \rangle < \epsilon_n -  \delta \}.
\end{align*}
Observe that
\begin{equation} \label{cover-lambda}
W_N \geq \lambda \ \ \ \hbox{if and only if} \ \ \ \lambda \mathbb{S}^{d-1} \not\subset \bigcup_{n=1}^N A_n(\lambda).
\end{equation}
It only remains to rescale \eqref{cover-lambda} to the unit-sphere.
Using the map $u \mapsto u /\lambda$ and the set $B_n(\lambda)$ defined by \eqref{Bn-def} one has that
\begin{equation} \label{coverage-equiv}
\lambda \mathbb{S}^{d-1} \not\subset \bigcup_{n=1}^N A_n(\lambda) \ \ \ \hbox{if and only if} \ \ \ \mathbb{S}^{d-1} \not\subset \bigcup_{n=1}^N B_n(\lambda).
\end{equation}
The proof of  \eqref{coverage-eq2} now follows by combining \eqref{cover-lambda} and \eqref{coverage-equiv}.
\end{proof}

We shall refer to the set $B_n(\lambda)$ as a bi-cap since it can be expressed as the union of two antipodal (possibly empty) open spherical caps
\begin{equation} \label{bicap-eq}
B_n(\lambda) =  {\rm Cap}(\varphi_n, \theta_n^{+})\cup  {\rm Cap}(-\varphi_n, \theta_n^{-}),
\end{equation}
where the angular radii $\theta_n^{+}$ and $\theta_n^{-}$ are given by 
\begin{equation} \label{thetaplus}
\theta^{+}_n =  
\begin{cases}
\arccos \left( \frac{\delta+\epsilon_n}{\lambda} \right), & \hbox{ if } \delta +  \epsilon_n < \lambda,\\
0, & \hbox{ otherwise,} 
\end{cases}
\end{equation}
and
\begin{equation} \label{thetaminus}
\theta^{-}_n = 
\begin{cases}
\arccos \left( \frac{\delta- \epsilon_n}{\lambda} \right), & \hbox{ if } \delta - \epsilon_n < \lambda,\\
0, & \hbox{ otherwise.} 
\end{cases}
\end{equation}
In particular, depending on the size of the parameters $\epsilon_n$ and $\lambda$, each set $B_n(\lambda)$ is either:
(i) a union of two disjoint spherical caps with antipodal centers, or (ii) a single spherical cap, or (iii) the empty set.

\subsection{Background on coverage processes} 
In our general analysis of consistent reconstruction, 
the coverage problem in Lemma \ref{cover-lem} involves spherical caps with both random angular radii and random centers.  
Random coverage problems have a long and technical history, e.g., see \cite{S, BCL10}, but in high dimensions the literature
is still limited when considering caps with both random center and random size.
In fact, even in the case of constant sized caps with random centers on $\mathbb{S}^{d-1}$, bounds on coverage probabilities were
only recently obtained in \cite{BCL10}.  
Some noteworthy results for randomly sized caps include \cite{SH} which contains exact results in dimension $d=2$ with general distributions on the random arclengths,
and \cite{J86} which contains asymptotic results (as the random cap size becomes small) on general manifolds.

In this section, we shall
provide some necessary background on coverage processes in the case of spherical caps with angular radii of fixed size and random centers $\{\varphi_n\}_{n=1}^N$
that are uniformly distributed on the unit-sphere.  The techniques and results that we will use later are especially influenced by \cite{BCL10} and \cite{FN77}.

For the remainder of this section let $\{\varphi_n\}_{n=1}^N \subset \mathbb{S}^{d-1}\subset \R^{d}$ be i.i.d. random vectors that are uniformly distributed on $\mathbb{S}^{d-1}$, 
and let $0<\theta< \pi/2$ be fixed.
Consider the following non-coverage probability
\begin{equation}\label{p-def}
p(N,d-1,\theta)  = \Pr \left( \mathbb{S}^{d-1} \not\subset \bigcup_{n=1}^N {\rm Cap}(\varphi_n,\theta) \right).
\end{equation}
The following theorem contains the best known bounds on $p(N,d-1,\theta)$ when $0<\theta<\pi/2$,
see Theorem 1.1 and Proposition 5.5 in \cite{BCL10}.   The work in \cite{BCL10} is stated for coverage by closed caps but the following result remains
true for open caps.

\begin{theorem}[B\"urgisser, Cucker, Lotz] \label{thmBCL}
If $0 < \theta < \pi/2$ and $N \geq d \geq 2$ then
\begin{align}
p(N,d-1,\theta) \leq 2^{1-N} \sum_{k=0}^{d-1} \binom{N-1}{k}
+ \binom{N}{d} \ \left( \frac{d\sqrt{d-1}}{2^{d-1}}  \right)  \ F_{N,d-1}(\theta), \label{BCL-bnd}
\end{align}
where
\begin{align*}
F_{N,d-1}(\theta)= \int_0^{\cos \theta} (1-t^2)^{((d-1)^2-2)/2}(1 - r_{d-1}(\arccos t))^{N-d-2} dt
\end{align*}
and $r_{d-1}( \arccos t)$ is defined using \eqref{rel-cap-meas}.
\end{theorem}

We briefly comment on why Theorem \ref{thmBCL} holds when $p(N,d-1,\theta)$ is defined as in \eqref{p-def} using open spherical caps instead 
of closed spherical caps as in \cite{BCL10}.  For this it suffices to note that if $0<\alpha <\theta<\pi/2$ then
$$\Pr \left( \mathbb{S}^{d-1} \not\subset \bigcup_{n=1}^N {\rm Cap}(\varphi_n,\theta) \right)
\leq \Pr \left( \mathbb{S}^{d-1} \not\subset \bigcup_{n=1}^N \overline{{\rm Cap}(\varphi_n,\alpha)} \right),$$
and that
$$\lim_{\alpha \to \theta} F_{N,d-1}(\alpha) = F_{N,d-1}(\theta).$$

We shall later need bounds on $p(N,d-1, \theta)$
when $\arccos(1/\sqrt{d}) \leq \theta < \pi/2$, i.e., when the cap height is less
than $1/\sqrt{d}$.

\begin{lemma} \label{pbound-lem}
If $N \geq \frac{2d}{\ln(12/11)}\approx (22.99)d$  and $\arccos(1/\sqrt{d}) \leq \theta < \pi/2$ then
\begin{align*}
p(N,(d-1), \theta) 
& \leq 2 \sqrt{d} \ \left( 13 \right)^d \left( \frac{11}{12} \right)^{N/2}.
\end{align*}
In particular, if 
$s> \frac{2 \ln(13)}{\ln(12/11)}\approx 58.96$ then
\begin{equation}
\lim_{d \to \infty} p(sd,(d-1), \theta)=0. \label{pds-lim}
\end{equation}
\end{lemma}

\begin{proof}
Hoeffding's inequality shows that if $(N-1) \geq 4(d-1)$ then
\begin{align}
2^{1-N} \sum_{k=0}^{d-1} \binom{N-1}{k} & \leq 
e^{1/8} e^{-N/8}.  \label{HemiBnd}
\end{align}
By Lemma 2.1 of  \cite{BGKKLS} we have that
$$\arccos(1/\sqrt{d}) \leq \theta < \pi/2 \implies \frac{1}{12} \leq r_{d-1}(\theta) \leq \frac{1}{2}.$$
Thus, for $\arccos(1/\sqrt{d}) \leq \theta < \pi/2$ we have
\begin{align}
F_{N,d-1}(\theta) &\leq \int_0^{1/\sqrt{d}} (1-t^2)^{((d-1)^2-2)/2}(1 - r_{d-1}(t))^{N-(d-1)-1} dt \notag  \\
&\leq \frac{1}{\sqrt{d}} \left( \frac{11}{12} \right)^{N-d}.  \label{F-bnd}
\end{align}
Combining \eqref{BCL-bnd}, \eqref{HemiBnd}, \eqref{F-bnd} 
gives
\begin{align}
p(N,d-1, \theta) &\leq e^{1/8}e^{-N/8} + \binom{N}{d} 
\frac{d\sqrt{d-1}}{2^{d-1}} \frac{1}{\sqrt{d}}  \left( \frac{11}{12} \right)^{N-d} \notag \\
& \leq e^{1/8}e^{-N/8} +  2d \left( \frac{12}{22} \right)^d \left( \frac{11}{12} \right)^{N}  \binom{N}{d}. \label{p-eq-unsimp}
\end{align}

We shall use the following bounds to further simplify \eqref{p-eq-unsimp}.  First, note that by Stirling's approximation
\begin{equation} \label{binom-stirling}
\binom{N}{d} \leq \frac{N^d}{d!} \leq \frac{1}{\sqrt{2\pi d}} \left( \frac{eN}{d} \right)^d.
\end{equation}
Also, it follows from $\ln(x) \leq (x/c) + (\ln(c) -1)$ that one has
\begin{equation} \label{dd-bnd}
\forall c>0, \ \ \ \left( \frac{eN}{d} \right)^d \leq c^d e^{eN/c} e^{-d}.
\end{equation}
Applying \eqref{binom-stirling} and \eqref{dd-bnd} with $c = 2e/ \ln(12/11)$ to \eqref{p-eq-unsimp} yields that
if $N \geq \frac{2d}{\ln(12/11)}$ then 
\begin{align*}
p(N,d-1, \theta) & \leq e^{1/8}e^{-N/8} +  2d \left( \frac{12}{22} \right)^d \left( \frac{11}{12} \right)^{N}  \binom{N}{d}\\
& \leq e^{1/8}e^{-N/8} + \frac{2\sqrt{d}}{\sqrt{2\pi}} \left( \frac{12}{22} \right)^d \left( \frac{2}{\ln(12/11)} \right)^d \left(\frac{11}{12} \right)^{N/2} \\
& \leq e^{1/8}e^{-N/8} + \sqrt{d} \thinspace \left( 13 \right)^d \left( \frac{11}{12} \right)^{N/2} \\
& \leq 2 \sqrt{d} \thinspace \left( 13 \right)^d \left( \frac{11}{12} \right)^{N/2}.
\end{align*}
This completes the proof.

\end{proof}

\section{Upper bounds for general distributions} \label{genthm-sec}

In this section we prove that consistent reconstruction achieves MSE of the optimal order $\mathbb{E}|W_N|^2 \lesssim 1/N^2$ under rather general conditions on the i.i.d. random
 measurement vectors $\{\varphi_n\}_{n=1}^N \subset \R^d$.  

 Our error bounds use the following admissibility condition. 
We assume that $\{\varphi_n\}_{n=1}^N\subset \R^d$ are independent identically distributed versions of a unit-norm random vector $\varphi\in \mathbb{S}^{d-1}$,
and we further assume that there exist constants $\alpha \geq 1, s>0$ such that
\begin{equation} \label{general-cond}
\forall \thinspace 0\leq t \leq 1, \ \forall x \in \mathbb{S}^{d-1}, \ \ \ \Pr \left( |\langle x, \varphi \rangle | \leq t \right) \leq \alpha \thinspace t^s.
\end{equation}
Roughly speaking, the admissibility condition \eqref{general-cond} ensures that the random vector $\varphi$ cannot be too concentrated on any subspace
of $\R^d$ with positive codimension.

\begin{example} \label{unif-gencond-ex}
If $d\geq 3$ and if the random vector $\varphi$ is uniformly distributed on the unit-sphere $\mathbb{S}^{d-1}$ then $\varphi$ satisfies \eqref{general-cond} with $s=1$ and $\alpha = 2C_d$, where $C_d$ is as in \eqref{CdDef}.  This follows since $0 \leq f_Z(z)\leq 2C_d$ in \eqref{Zpdf}.   Similarly, when $d=2$, a direct computation shows that if $\varphi$ is uniformly distributed on the unit-circle $\mathbb{S}^{1}$ then $\varphi$ satisfies \eqref{general-cond} with $s=1$ and $\alpha = 1$.  
This follows using $2C_2 = 2/\pi$ and $\arcsin(t) \leq (\pi/2)t$.
\end{example}

\begin{example} \label{beta-dominate-ex}
Suppose that the unit-norm random vector $\varphi_1\in \mathbb{S}^{d-1}$ satisfies \eqref{general-cond} with $\alpha=\alpha_1$ and $s=s_1$.
Suppose that the unit-norm random vector $\varphi_2 \in \mathbb{S}^{d-1}$ has the property that there exists $\beta>0$ such that
for every Borel subset $B\subset \mathbb{S}^{d-1}$ there holds
$\Pr [ \varphi_2\in B] \leq \beta \thinspace \Pr [ \varphi_1\in B].$
Then $\varphi_2$ satisfies \eqref{general-cond} with $\alpha=\beta \alpha_1$ and $s=s_1$.
\end{example}

\begin{example}
If the random vector $\varphi$ is uniformly distributed on an open subset of $\mathbb{S}^{d-1}$
then $\varphi$ satisfies \eqref{general-cond}.  This follows by combining Examples \ref{unif-gencond-ex} and \ref{beta-dominate-ex}.
\end{example}

\begin{example}
If the random vector $\varphi$ has a point mass at $v\in \R^d$ with $\Pr [ \varphi = v]>0$ then $\varphi$ does not
satisfy \eqref{general-cond}. To see this, let $x_v\in \mathbb{S}^{d-1}$ be any vector that is orthogonal to $v$.
Then $\Pr [ \langle \varphi, x_v \rangle  =0] \geq \Pr [ \varphi =v] >0$ shows that \eqref{general-cond} does not hold.
\end{example}

We are now ready to state and prove our first main theorem.
\begin{theorem} \label{main-thm1}
Suppose that $\{\varphi_n\}_{n=1}^N \subset \mathbb{S}^{d-1}$ are i.i.d. versions of a unit-norm random vector $\varphi$ that satisfies \eqref{general-cond}.
For all $N \geq (d+2)/s$ the expected worst case error squared satisfies
\begin{equation} \label{thm1bnd}
\mathbb{E} |W_N|^2 \leq \frac{10^5 \delta^2d^2(2 \alpha)^{2/s}\ln^2(16 (2\alpha)^{1/s})}{(N+1)(N+2)}
+  \delta^2 32^{d+1}(2\alpha)^{(d+1)/s}\left(\frac{1}{2} \right)^N.
\end{equation}
\end{theorem}

\begin{proof}  The proof is divided into several steps.

\noindent {\em Step I.}  To use the error represention \eqref{MSE-intbyparts} we need to bound $\Pr [ W_N > \lambda]$.  By Lemma \ref{cover-lem}
this will be done by bounding the coverage probability \eqref{coverage-eq}.  
We begin by discretizing the coverage problem \eqref{coverage-eq} with an $\epsilon$-net argument developed in \cite{FN77}.

Given any $\epsilon>0$, let $\mathcal{N}_{\epsilon} = \{z_m\}_{m=1}^M \subset \mathbb{S}^{d-1}$ be a geodesic $\epsilon$-net for $\mathbb{S}^{d-1}$ with cardinality satisfying
$$M=\#(\mathcal{N}_{\epsilon}) \leq \left( \frac{8}{\epsilon} \right)^{d-1}.$$
Recall that the bi-cap $B_n(\lambda)$ is defined by  \eqref{Bn-def} and \eqref{bicap-eq} as
$$B_n(\lambda) =  {\rm Cap}(\varphi_n, \theta_n^{+})\cup  {\rm Cap}(-\varphi_n, \theta_n^{-}).$$
Next define the shrunken bi-cap 
\begin{equation} 
T_{\epsilon}(B_n(\lambda)) = {\rm Cap}(\varphi_n, T_{\epsilon}(\theta_n^{+}))\cup  {\rm Cap}(-\varphi_n, T_{\epsilon}(\theta_n^{-})),
\end{equation}
where $T_{\epsilon}(\theta_n^{+})$ and $T_{\epsilon}(\theta_n^{-})$ are defined by  \eqref{thetaplus}, \eqref{thetaminus} and \eqref{soft-thresh}.

The key discretization step is to proceed as in \cite{FN77} and note that
\begin{equation} \label{cover-discrete}
\mathbb{S}^{d-1} \not\subset \bigcup_{n=1}^N B_n(\lambda) \ \ \ \implies \mathcal{N}_{\epsilon} \not\subset \bigcup_{n=1}^N T_{\epsilon}(B_n(\lambda)).
\end{equation}
Since the shrunken bi-caps $\{T_{\epsilon}(B_n(\lambda))\}_{n=1}^N$ are independent and identically distributed, \eqref{cover-discrete} implies
\begin{align}
\Pr ( W_N > \lambda ) &\leq \Pr \left( \mathbb{S}^{d-1} \not\subset \bigcup_{n=1}^N B_n(\lambda) \right) \notag \\
&\leq \Pr \left(  \mathcal{N}_{\epsilon} \not\subset \bigcup_{n=1}^N T_{\epsilon}(B_n(\lambda)) \right) \notag\\
& \leq \sum_{m=1}^{M}  \Pr \left( z_m \notin \bigcup_{n=1}^N T_{\epsilon}(B_n(\lambda)) \right) \notag \\
& = \sum_{m=1}^{M} \large\left(  \Pr \left[ z_m \notin  T_{\epsilon}(B_1(\lambda)) \right] \large\right)^N \notag\\
& \leq \left( \frac{8}{\epsilon} \right)^{d-1} \sup_{z\in \mathbb{S}^{d-1}} \left( \Pr [ z \notin T_{\epsilon}(B_1(\lambda)) ] \right)^N.\label{cover-discrete-bnd}\\ \notag
\end{align}

\noindent {\em Step II.}  We now use \eqref{cover-discrete-bnd} to bound $\Pr[W_N>\lambda]$ in the case when $\lambda \geq 4 \delta$.
In this case note that each $B_n(\lambda)$ is a genuine bi-cap that consists of two antipodal non-empty spherical caps.
Since $\lambda >2 \delta$ it is straightforward from \eqref{Bn-def} that
$$B_n(\lambda) \supset  \left\{ u \in \mathbb{S}^{d-1} : \langle u , \varphi_n \rangle > \frac{2\delta}{\lambda} \  
\hbox{ or } \  \langle u , \varphi_n \rangle < \frac{-2\delta}{\lambda}  \right\},$$
and it follows that the shrunken bi-cap $T_{\epsilon}(B_n(\lambda))$ satisfies
\begin{align} 
T_{\epsilon}(B_n(\lambda)) &\supset 
\left\{ u \in \mathbb{S}^{d-1} : |\langle u , \varphi_n \rangle|  >  \cos \left( \arccos \left(\frac{2\delta}{\lambda}\right) -\epsilon \right) \right\} \notag \\
& \supset  \left\{ u \in \mathbb{S}^{d-1} : |\langle u , \varphi_n \rangle| > \frac{2\delta}{\lambda} +\epsilon  \right\}. \label{bicap-subset}
\end{align}

For the remainder of this step we fix $\epsilon =  \frac{2\delta}{\lambda}$.  The assumption \eqref{general-cond}  along with \eqref{bicap-subset}
shows that if $\lambda \geq 4\delta$ then
\begin{align}
\Pr [ z \notin T_{\epsilon}(B_1(\lambda))] &\leq \Pr \left( |\langle z, \varphi \rangle| \leq\frac{2\delta}{\lambda} + \epsilon \right) \label{eq-useinthm2}\\
& = \Pr \left( |\langle z, \varphi \rangle| \leq\frac{4\delta}{\lambda} \right) \notag\\
& \leq \alpha \left( \frac{ 4\delta }{\lambda} \right)^s. \label{step2bnd}
\end{align}
Thus \eqref{cover-discrete-bnd} and \eqref{step2bnd} imply that if $\lambda \geq 4\delta$ then
\begin{equation}  \label{step2result}
 \Pr [ W_N > \lambda ] \leq \left( \frac{4\lambda}{\delta} \right)^{d-1}  \left( \alpha \left( \frac{ 4\delta }{\lambda} \right)^s \right)^N
 = 4^{d-1} (4^s \alpha)^N \left( \frac{\delta}{\lambda}\right)^{sN-d+1}.\\
\end{equation}
Since $8\delta(2\alpha)^{1/s} \geq 4\delta$ note that if $N \geq (d+2)/s$ then \eqref{step2result} implies
\begin{align}
\int_{8\delta(2\alpha)^{1/s}}^{\infty} \lambda \Pr [ W_N > \lambda] d \lambda  
& \leq 4^{d-1} (4^s \alpha)^N \int_{8\delta(2\alpha)^{1/s}}^{\infty} \lambda  \left( \frac{\delta}{\lambda}\right)^{sN-d+1} d \lambda \notag\\
& = 4^{d-1} \left(\frac{1}{2(2^s)} \right)^N \delta^2 \frac{(8(2\alpha)^{1/s})^{d+1}}{(sN-d-1)} \notag \\
& \leq \delta^2 (32(2\alpha)^{1/s})^{d+1} \left(\frac{1}{2} \right)^N. \label{int2}
\end{align}

\noindent {\em Step III.}  Next, we bound $\Pr[W_N>\lambda]$ in the case when $0< \lambda \leq {4(2\alpha)^{1/s}\delta}$.

It will be useful to begin with the following symmetrization argument.  
Let $\{b_n\}_{n=1}^N$ be i.i.d. Bernoulli random variables satisfying $\Pr[b_n=1] = \Pr[b_n=-1]=1/2$, and additionally 
suppose that $\{b_n\}_{n=1}^N$ is independent of $\{\varphi_n\}_{n=1}^N$ and $\{\epsilon_n\}_{n=1}^N$.  
Note that the i.i.d. bi-caps $\{B_n(\lambda)\}_{n=1}^N$
have the same distribution as the i.i.d. bi-caps $\{B_n^{\prime}(\lambda)\}_{n=1}^N$ defined by
$$B_n^{\prime}(\lambda) = {\rm Cap}(b_n \varphi_n, \theta^+_n(\lambda))\cup {\rm Cap}(-b_n \varphi_n, \theta^-_n(\lambda)).$$
This follows from \eqref{bicap-eq}, \eqref{thetaplus}, \eqref{thetaminus} and the fact that $\epsilon_n$ is uniformly distributed on $[-\delta,\delta]$.

Consequently,
$$\Pr[\mathbb{S}^{d-1}  \not\subset \bigcup_{n=1}^N B_n(\lambda)] =  \Pr[ \mathbb{S}^{d-1}  \not\subset \bigcup_{n=1}^N B^{\prime}_n(\lambda)].$$
Let $I(\lambda)$ denote the number of $\{\epsilon_n + \delta\}_{n=1}^N$ that lie in the interval $\left[0,\frac{\lambda}{2(2\alpha)^{1/s}} \right] $, namely:
$$I(\lambda) = \# \left( \{\epsilon_n + \delta\}_{n=1}^N \cap \left[0,\frac{\lambda}{2(2\alpha)^{1/s}} \right] \right).$$
Let
\begin{equation} \label{qdef}
q(k, d-1, \alpha,s)=\Pr \left[ \mathbb{S}^{d-1}  \not\subset \bigcup_{n=1}^k {\rm Cap} \left(b_n \varphi_n, \arccos\left(\frac{1}{2(2\alpha)^{1/s}}\right) \right)\right].
\end{equation}
Thus
\begin{align}
\Pr [ W_N>\lambda] &\leq \Pr[ \mathbb{S}^{d-1}  \not\subset \bigcup_{n=1}^N B_n(\lambda)] \notag \\
& = \Pr[ \mathbb{S}^{d-1}  \not\subset \bigcup_{n=1}^N B^{\prime}_n(\lambda)] \notag \\
& \leq \Pr[ \mathbb{S}^{d-1}  \not\subset  \bigcup_{n=1}^N{\rm Cap}(b_n \varphi_n, \theta^+_n(\lambda))] \notag \\
& = \sum_{k=0}^N \Pr \left[ \mathbb{S}^{d-1}  \not\subset \bigcup_{n=1}^N {\rm Cap}(b_n \varphi_n, \theta^+_n(\lambda) ) \ \bigg\vert \ I(\lambda) = k\right] \ \Pr [ I(\lambda)=k] \notag \\
& \leq \sum_{k=0}^N \Pr \left[ \mathbb{S}^{d-1}  \not\subset \bigcup_{n=1}^k {\rm Cap}(b_n \varphi_n, \arccos(1/2(2\alpha)^{1/s}) )\right] \ \Pr [ I(\lambda)=k] \notag \\
& = \sum_{k=0}^N q(k, d-1, \alpha,s) \ \Pr[ I(\lambda) = k] \notag \\
& = \sum_{k=0}^N q(k, d-1, \alpha,s) \binom{N}{k} \left( 1 - \frac{\lambda}{4\delta(2\alpha)^{1/s}}\right)^{N-k} \left( \frac{\lambda}{4\delta(2\alpha)^{1/s}} \right)^k. \label{thm1step3eq}
\end{align}

\noindent {\em Step IV.}  In this step we bound the integral $\int_0^{4\delta (2\alpha)^{1/s}} \lambda \Pr[W_N > \lambda] d\lambda$.  
Equation \eqref{thm1step3eq} and properties of the beta function imply
\begin{align}
&\int_0^{4\delta(2\alpha)^{1/s}}  \lambda \Pr[W_N > \lambda] d\lambda \notag \\
& \leq \sum_{k=0}^N q(k, d-1, \alpha,s) \binom{N}{k}  
\int_0^{4\delta (2\alpha)^{1/s}} \lambda \left( 1 - \frac{\lambda}{4\delta(2\alpha)^{1/s}} \right)^{N-k} \left( \frac{\lambda}{4\delta(2\alpha)^{1/s}} \right)^k d \lambda \notag\\
&= \sum_{k=0}^N q(k, d-1, \alpha,s) \binom{N}{k} {16\delta^2(2\alpha)^{2/s}} \int_0^1 u^{k+1} (1- u)^{N-k} d \lambda \notag \\
&=\sum_{k=0}^N q(k, d-1, \alpha,s) \binom{N}{k}  {16\delta^2(2\alpha)^{2/s}}  \frac{(k+1)! (N-k)!}{(N+2)!} \notag \\
&=\frac{16\delta^2(2\alpha)^{2/s}}{(N+1)(N+2)} \sum_{k=0}^N q(k, d-1, \alpha,s) \thinspace (k+1) \notag \\
&\leq \frac{16\delta^2(2\alpha)^{2/s}}{(N+1)(N+2)} \left( 1 + 2 \sum_{k=1}^N k \thinspace q(k, d-1,\alpha,s) \right). \label{thm1stepivqsum}
\end{align}

\noindent {\em Step V.}  In this step we bound the quantity $q(k, d-1, \alpha,s)$ appearing in \eqref{thm1stepivqsum}
and defined in \eqref{qdef}.
We shall again employ an $\epsilon$-net argument as in Step I.  
For the remainder of this step fix 
\begin{equation} \label{thm1stepvepsdef}
\epsilon = \frac{1}{2(2 \alpha)^{1/s}} \left( 1 - \frac{1}{4(2\alpha)^{2/s}}\right)^{-1/2},
\end{equation}
and let $\mathcal{N}_{\epsilon}$ be a geodesic $\epsilon$-net for $\mathbb{S}^{d-1}$ with cardinality $\#(\mathcal{N}_{\epsilon}) \leq \left( \frac{8}{\epsilon} \right)^{d-1}$.
A similar argument as used to obtain \eqref{cover-discrete-bnd} yields the following
\begin{align}
q(k, d-1,\alpha,s)  & = \Pr \left[ \mathbb{S}^{d-1} \not\subset \bigcup_{n=1}^k {\rm Cap}(b_n \varphi_n, \arccos(1/2(2\alpha)^{1/s}) ) \right]  \notag \\
& \leq \Pr [ \mathcal{N}_{\epsilon} \not\subset \bigcup_{n=1}^k {\rm Cap}(b_n \varphi_n, \arccos(1/2(2\alpha)^{1/s}) - \epsilon )] \notag \\
&\leq \left( \frac{8}{\epsilon} \right)^{d-1} 
\left( \sup_{z\in \mathbb{S}^{d-1}} \Pr \left[ z \notin {\rm Cap}(b_1 \varphi_1, \arccos(1/2(2\alpha)^{1/s}) - \epsilon ) \right]  \right)^k. \label{thm1stepvqepsnetbnd}
\end{align}
By \eqref{general-cond} and \eqref{thm1stepvepsdef}, for an arbitrary $z\in \mathbb{S}^{d-1}$ we may compute as follows
\begin{align}
\Pr &\left[ z \notin {\rm Cap}(b_1 \varphi_1, \arccos(1/2(2\alpha)^{1/s})- \epsilon ) \right] \notag \\
&= \frac{1}{2} \Pr \left[ \langle z, \varphi_1 \rangle \leq \cos(\arccos(1/ 2(2\alpha)^{1/s}) - \epsilon) \right]
+  \frac{1}{2} \Pr \left[ -\langle z, \varphi_1 \rangle \leq \cos(\arccos(1/ 2(2\alpha)^{1/s}) - \epsilon) \right] \notag \\
&= \frac{1}{2} \left( 1 + \Pr\left[ |\langle z, \varphi_1 \rangle| \leq \cos(\arccos(1/ 2(2\alpha)^{1/s}) - \epsilon) \right] \right) \notag\\
&=  \frac{1}{2} \left( 1 + \Pr\left[ |\langle z, \varphi_1 \rangle| \leq \frac{\cos \epsilon}{2(2\alpha)^{1/s}}  
+ \left(1 - \frac{1}{4(2\alpha)^{2/s}}\right)^{1/2} \sin \epsilon  \right] \right) \notag \\
&\leq \frac{1}{2} \left( 1 + \Pr\left[ |\langle z, \varphi_1 \rangle| \leq \frac{1}{2(2\alpha)^{1/s}}  + \left(1 - \frac{1}{4(2\alpha)^{2/s}}\right)^{1/2}  \epsilon  \right] \right) \notag \\
&= \frac{1}{2} \left( 1 + \Pr\left[ |\langle z, \varphi_1 \rangle| \leq \frac{1}{(2\alpha)^{1/s}}   \right] \right) \notag \\
& \leq  \frac{1}{2} \left( 1 + \alpha \left|\frac{1}{(2\alpha)^{1/s}}\right|^s \right) \notag \\
&= \left(\frac{3}{4}\right). 
\label{thm1stepv34bnd}
\end{align}
Combining \eqref{thm1stepvepsdef}, \eqref{thm1stepvqepsnetbnd}, and \eqref{thm1stepv34bnd} gives
\begin{align}
q(k,d-1,\alpha,s)
&\leq \left( \frac{8}{\epsilon} \right)^{d-1}  \left( \frac{3}{4} \right)^k \notag \\
& =  \left( 8 \sqrt{4(2\alpha)^{2/s}-1} \right)^{d-1} \left( \frac{3}{4} \right)^k \notag\\
& \leq \left( 16 (2\alpha)^{1/s} \right)^{d} \left( \frac{3}{4} \right)^k. \label{StepVeq}
\end{align}

\vspace{.1in}
\noindent {\em Step VI.}  In this step we bound the sum 
$\sum_{k=1}^N k \thinspace q(k, d-1, \alpha,s)$ appearing in \eqref{thm1stepivqsum}.
For the remainder of this step let 
\begin{equation} \label{stepviKdef}
K = \frac{2d \ln(16 (2\alpha)^{1/s})}{\ln(4/3)},
\end{equation}
and note that $K>1$.
If $k \geq K$ then by \eqref{StepVeq}
$$q(k,d-1,\alpha,s) \leq \left(\frac{3}{4}\right)^{k/2}.$$
Thus
\begin{align}
\sum_{k=K+1}^N k \thinspace q(k, d-1,\alpha,s)  
&\leq \sum_{k=K+1}^N k \left(\frac{3}{4}\right)^{k/2} \notag \\
&\leq \sum_{k=1}^{\infty} (k+K) \left(\sqrt{\frac{3}{4}}\right)^{k+K}\notag \\
&\leq 2K \sum_{k=1}^{\infty} k \left(\sqrt{\frac{3}{4}}\right)^{k}\notag \\
&\leq 120K. \label{Kupsum}
\end{align}
Here we have used that
\begin{equation} \label{geosum}
\forall \thinspace 0< r< 1, \ \ \ \sum_{k=1}^{\infty} k r^k = \frac{1}{(1-r)^2}.
\end{equation}
Since $0 \leq q(k, d-1,\alpha,s) \leq 1$ we also have
\begin{align}
\sum_{k=1}^{K} k \thinspace q(k, d-1, \alpha,s) & \leq
\sum_{k=1}^{K} k \leq K^2. \label{Klowsum}
\end{align}
Since $K>1$, combining \eqref{Kupsum} and \eqref{Klowsum} gives
\begin{equation} \label{stepviqsumeq}
1+2\sum_{k=1}^N k q(k, d-1,\alpha,s) \leq 1+2(120K+K^2) \leq 250K^2.
\end{equation}
By \eqref{thm1stepivqsum}, \eqref{stepviKdef}, and \eqref{stepviqsumeq} we conclude that
\begin{align}
\int_0^{4\delta(2\alpha)^{1/s}} \lambda \Pr[W_N> \lambda] d \lambda 
&\leq \frac{16 \delta^2(2 \alpha)^{2/s} (250K^2)}{(N+1)(N+2)} \notag \\
&\leq\frac{(16)(250) \delta^2(2 \alpha)^{2/s}}{(N+1)(N+2)} \left(\frac{2d \ln(16 (2\alpha)^{1/s})}{\ln(4/3)}\right)^2 \notag \\
& \leq\frac{10^5 \delta^2d^2(2 \alpha)^{2/s}\ln^2(16 (2\alpha)^{1/s})}{(N+1)(N+2)}. \label{thm1stepvieq}
\end{align}

\noindent {\em Step VII.}   It remains to bound the integral $\int_{4\delta(2\alpha)^{1/s}}^{8\delta(2\alpha)^{1/s}} \lambda \Pr[W_N>\lambda] d \lambda$.
Note that
$$4\delta(2\alpha)^{1/s} \leq \lambda \implies \Pr[W_N>\lambda] \leq \Pr[W_N>4\delta(2\alpha)^{1/s}].$$
Since $4\delta(2\alpha)^{1/s}\geq 4\delta$ it follows from \eqref{step2result} that
\begin{align*}
 \Pr[W_N>4\delta(2\alpha)^{1/s}] & \leq 4^{d-1}(4(2\alpha)^{1/s})^{d-1} \left( \frac{1}{2}\right)^N.
\end{align*}
Consequently
\begin{align}
\int_{4\delta(2\alpha)^{1/s}}^{8\delta(2\alpha)^{1/s}} \lambda \Pr[W_N>\lambda] d \lambda 
&\leq 4^{d-1}(4(2\alpha)^{1/s})^{d-1} \left( \frac{1}{2}\right)^N 24 \delta^2 (2\alpha)^{2/s} \notag \\
& \leq 2\delta^2 16^d (2\alpha)^{(d+1)/s} \left( \frac{1}{2}\right)^N. \label{thm1stepviieq}
\end{align}
The bound \eqref{thm1bnd} now follows from \eqref{int2}, \eqref{thm1stepvieq}, and \eqref{thm1stepviieq}.  This completes the proof.
\end{proof}

\begin{corollary} 
Suppose that $\{\varphi_n\}_{n=1}^N \subset \mathbb{S}^{d-1}$ are i.i.d. versions of a unit-norm random vector $\varphi$ that satisfies \eqref{general-cond}.
There exist absolute constants $c_1, c_2>0$ such that if 
$$N \geq c_2  d  \ln (32(2\alpha)^{1/s})$$ then the expected worst case error squared satisfies
\begin{equation} \label{corthm1bnd}
\mathbb{E}|W_N|^2 \leq \frac{c_1 \delta^2 d^2(2\alpha)^{2/s} \ln^2(16(2\alpha)^{1/s}) }{(N+1)(N+2)}.
\end{equation}
\end{corollary}

\begin{example}[Uniformly distributed measurements]
If $d\geq3$ and $\{\varphi_n\}_{n=1}^N\subset \mathbb{S}^{d-1}$ are i.i.d. uniformly distributed
random vectors on $\mathbb{S}^{d-1}$ then Example \ref{unif-gencond-ex} and \eqref{C-d-estimate} show that one can take 
$s=1$ and $\alpha = 2C_d \leq 2\sqrt{\frac{d-1}{2\pi}}\leq \sqrt{d}$.
In this case the error bound \eqref{corthm1bnd} shows that there exist constants $a,c$ such that if  $N \geq a d \ln d$ then
\begin{equation}
\mathbb{E}|W_N|^2 \leq \frac{c \thinspace \delta^2 d^3 \ln^2 d}{(N+1)(N+2)}. \label{d3logd}
\end{equation}
In the next section (see Theorem \ref{mainthm2}) we show that the logarithmic term in \eqref{d3logd} can in fact be removed.
\end{example}

For perspective, note that the admissibilty condition \eqref{general-cond} is stronger than the condition \eqref{rg-exp-cond} that was used to analyze the Rangan-Goyal algorithm.
In particular, it is straightforward to show that if $\varphi \in \mathbb{S}^{d-1}$ is a random vector that satisfies \eqref{general-cond}, then
$$\forall x \in \mathbb{S}^{d-1}, \ \mathbb{E} |\langle x, \varphi \rangle| \geq \left(\frac{1}{\alpha}\right)^s \left( 1- \frac{1}{s+1} \right)>0.$$

We conclude this section by noting that Theorem \ref{main-thm1} does not generally hold under the weaker condition \eqref{rg-exp-cond}.
For this we first show if $\varphi$ has a point mass then the conclusion \eqref{thm1bnd} does not hold.
\begin{example}\label{atom-counter-examp}
Suppose $\{ \varphi_n\}_{n=1}^N \subset \R^d$ are i.i.d. versions of a random vector $\varphi$ that has  a point mass at $v \in \R^d$ that occurs
with positive probability $c=\Pr [ \varphi = v]>0$.   Then
$$\forall \lambda >0, \ \ \ \Pr [W_N > \lambda ] \geq \Pr [W_N = \infty] \geq \Pr [ \varphi_1 = \varphi_2 = \cdots = \varphi_N = v] = c^N.$$
Thus
$$\mathbb{E}|W_N|^2 = 2 \int_0^{\infty} \lambda \Pr [W_N > \lambda ] d \lambda \geq 2 \int_0^{\infty} \lambda c^N  d \lambda = \infty.$$
In particular, the conclusion \eqref{thm1bnd} of Theorem \ref{main-thm1} fails if $\varphi$ has a point mass.
\end{example}

\begin{example}
Let $\{e_n\}_{n=1}^d \subset \R^d$ be an orthonormal basis.  Let $\varphi$ be the discrete random vector defined by
$\Pr [ \varphi = e_n] = 1/d$ for each $1\leq n \leq d$.  Then
$$\forall x \in \R^d, \ \ \ \mathbb{E} | \langle x, \varphi \rangle | = (1/d)\sum_{n=1}^d | \langle x, e_n \rangle | \geq (1/d) \left( \sum_{n=1}^d | \langle x, e_n \rangle |^2 \right)^{1/2}
= (1/d)\|x\|.$$
So, $\varphi$ satisfies \eqref{rg-exp-cond}, but by Example \ref{atom-counter-examp}, the conclusion \eqref{thm1bnd} of
Theorem \ref{main-thm1} does not hold.
\end{example}

\section{Upper bounds for uniformly distributed measurements} \label{unif-sec}

In this section we prove refined bounds for consistent reconstruction when $\{\varphi_n\}_{n=1}^N\subset \mathbb{S}^{d-1}$ are i.i.d.
random vectors that are uniformly distributed on the unit-sphere.  In this case, our next main result shows that
the dimension dependent constant $K$ in \eqref{overview-eq} is essentially dominated by the cube of the ambient dimension.

\begin{theorem} \label{mainthm2}
If the random vectors $\{\varphi_n\}_{n=1}^N \subset \R^d$ are i.i.d. uniformly distributed on $\mathbb{S}^{d-1}$ and if $N \geq d+2$ then
\begin{equation}
\mathbb{E}|W_N|^2 \leq \frac{2e^{12}\delta^2 d^3}{(N+1)(N+2)} + 26\delta^2 d^{3/2}\left( \frac{11}{12} \right)^{N/2} e^{\frac{d \ln(1024d)}{2}}.
\end{equation}
In particular, 
$$\limsup_{d\to \infty} \mathbb{E}|W_{\lceil d^{3/2} \rceil}|^2 \leq 2e^{12}\delta^2 \ \ \ \hbox{ and } \ \ \ 
 \lim_{d\to \infty} \mathbb{E}|W_{ \lceil d^{3/2} \ln d \thinspace \rceil}|^2 = 0.$$
\end{theorem}

\begin{proof}  The proof is divided into several steps.  We need to bound the integral
\begin{equation} \label{thm2-eq1}
\mathbb{E}|W_N|^2 = 2 \int_0^{\infty} \lambda \Pr [ W_N > \lambda ] d \lambda.
\end{equation}

\noindent {\em Step I.}   In this step we provide preliminary bounds on $\Pr [ W_N > \lambda ]$ when $0 \leq \lambda \leq 2\delta\sqrt{d}$.

Note that $\{\epsilon_n + \delta\}_{n=1}^N$ are i.i.d. uniform random variables on $[0, 2\delta]$ and define the random variable
$$J(\lambda) = \# \left( \{ \epsilon_n+\delta \}_{n=1}^N \cap \left[0, \frac{\lambda}{\sqrt{d}}\right] \right).$$
Lemma \ref{cover-lem}, \eqref{bicap-eq}, \eqref{thetaplus} imply that
\begin{align}
\Pr [ W_N>\lambda] &\leq \Pr[ \mathbb{S}^{d-1}  \not\subset \bigcup_{n=1}^N B_n(\lambda)] \notag \\
& \leq \Pr[ \mathbb{S}^{d-1}  \not\subset  \bigcup_{n=1}^N{\rm Cap}(\varphi_n, \theta^+_n(\lambda))] \notag \\
& = \sum_{k=0}^N \Pr \left[ \mathbb{S}^{d-1}  \not\subset \bigcup_{n=1}^N {\rm Cap}(\varphi_n, \theta^+_n(\lambda) ) \ \bigg\vert \ J(\lambda) = k\right] \ \Pr [ J(\lambda)=k] \notag \\
& \leq \sum_{k=0}^N \Pr \left[ \mathbb{S}^{d-1}  \not\subset \bigcup_{n=1}^k {\rm Cap}(\varphi_n, \arccos(1/\sqrt{d}) )\right] \ \Pr [ J(\lambda)=k] \notag \\
& = \sum_{k=0}^N p(k, d-1, \arccos(1/\sqrt{d})) \ \Pr[ J(\lambda) = k] \notag \\
& = \sum_{k=0}^N p(k, d-1, \arccos(1/\sqrt{d})) \binom{N}{k} \left( 1 - \frac{\lambda}{2\delta\sqrt{d}}\right)^{N-k} \left( \frac{\lambda}{2\delta\sqrt{d}} \right)^k \label{thm2-eq2}
\end{align}
The above computations made use of definition \eqref{p-def} and the independence of the $\{\epsilon_n\}_{n=1}^N$ and $\{\varphi_n\}_{n=1}^N$.\\

\noindent {\em Step II.}  In this step we bound the integral $\int_0^{2\delta \sqrt{d}} \lambda \Pr[W_N > \lambda] d\lambda$.  Equation \eqref{thm2-eq2} and properties of the beta function imply
\begin{align}
\int_0^{2\delta\sqrt{d}} & \lambda \Pr[W_N > \lambda] d\lambda \notag \\
& \leq \sum_{k=0}^N p(k, d-1, \arccos(1/\sqrt{d})) \binom{N}{k}  \int_0^{2\delta\sqrt{d}} \lambda \left( 1 - \frac{\lambda}{2\delta\sqrt{d}}\right)^{N-k} \left( \frac{\lambda}{2\delta\sqrt{d}} \right)^k d \lambda \notag\\
&= \sum_{k=0}^N p(k, d-1, \arccos(1/\sqrt{d})) \binom{N}{k} (4\delta^2 d) \int_0^1 u^{k+1} (1- u)^{N-k} d \lambda \notag \\
&=\sum_{k=0}^N p(k, d-1, \arccos(1/\sqrt{d})) \binom{N}{k}  (4\delta^2 d) \frac{(k+1)! (N-k)!}{(N+2)!} \notag \\
&=\frac{4\delta^2 d}{(N+1)(N+2)} \sum_{k=0}^N p(k, d-1, \arccos(1/\sqrt{d})) \thinspace (k+1) \notag \\
&\leq \frac{4\delta^2 d}{(N+1)(N+2)} \left( 1 + 2 \sum_{k=1}^N k \thinspace p(k, d-1, \arccos(1/\sqrt{d})) \right). \label{intsumsqrtd}
\end{align}

\noindent {\em Step III.}  In this step we bound the sum $\sum_{k=1}^N k \thinspace p(k, d-1, \arccos(1/\sqrt{d}))$ appearing in \eqref{intsumsqrtd}.

First note that
\begin{equation}
\sum_{k=1}^{60d} k \thinspace p(k, d-1, \arccos(1/\sqrt{d})) \leq \sum_{k=1}^{60d} k = \frac{60d(60d+1)}{2} \leq 3600d^2. \label{thm2-eq3}
\end{equation}
Next by Lemma \ref{pbound-lem} and \eqref{geosum}
\begin{align}
\sum_{k=60d+1}^N k \thinspace p(k, d-1, \arccos(1/\sqrt{d})) & \leq \sum_{k=60d+1}^N k \left( 2 \sqrt{d} (13)^d \left( \frac{11}{12} \right)^{k/2} \right)\notag \\
& = 2 \sqrt{d} (13)^d \sum_{k=60d+1}^N k\left( \frac{11}{12} \right)^{k/2}\notag \\
& \leq 2 \sqrt{d} (13)^d \sum_{k=1}^{\infty}  (k+60d) \left( \sqrt{\frac{11}{12}} \right)^{k+60d}\notag \\
& \leq 2 \sqrt{d} (120d) \sum_{k=1}^{\infty} k \left( \sqrt{\frac{11}{12}} \right)^{k} \notag \\
& = 2 \sqrt{d} (120d) \left( \frac{\sqrt{132}}{(\sqrt{12} - \sqrt{11})^2} \right) \notag \\
& \leq  12700d^{3/2}.
\label{thm2-eq4}
\end{align}
Combining \eqref{thm2-eq3} and \eqref{thm2-eq4} shows that 
\begin{equation} \label{thm2-eq5}
\sum_{k=1}^N  k \thinspace p(k, d-1, 1/\sqrt{d}) \leq 3600d^2 + 12700d^{3/2}.
\end{equation}
Now \eqref{intsumsqrtd} and \eqref{thm2-eq5} show that
\begin{equation}
\int_0^{2\delta \sqrt{d}} \lambda  \Pr[W_N > \lambda] d\lambda \leq \frac{e^{12}\delta^2d^3}{(N+1)(N+2)}. \label{thm2-eq5b}\\
\end{equation}

\noindent {\em Step IV.}  In this step we bound the integral $\int_{2\delta\sqrt{d}}^{4\delta\sqrt{d}} \lambda \Pr[W_N>\lambda] d \lambda$.
Note that if $\lambda \geq 2\delta \sqrt{d}$ then \eqref{thetaplus} implies that
$\theta^+_n(\lambda) \geq \arccos(1/\sqrt{d})$ which, in turn, implies that
\begin{equation} \label{thm2-eq6}
{\rm Cap}(\varphi_n, \arccos(1/\sqrt{d})) \subset  {\rm Cap}(\varphi_n, \theta^+_n(\lambda))
\end{equation}
Lemma \ref{cover-lem} and equations \eqref{bicap-eq} and \eqref{thm2-eq6} imply that
\begin{align*}
\Pr [W_N>\lambda] &= \Pr [ \mathbb{S}^{d-1} \not\subset \bigcup_{n=1}^N B_n(\lambda)] \\
&\leq \Pr [ \mathbb{S}^{d-1} \not\subset \bigcup_{n=1}^N {\rm Cap}(\varphi_n, \theta^+_n(\lambda))]\\
&\leq \Pr [ \mathbb{S}^{d-1} \not\subset \bigcup_{n=1}^N {\rm Cap}(\varphi_n, \arccos(1/\sqrt{d}))]\\
&= p(N, d-1, \arccos(1/\sqrt{d})).
\end{align*}
Thus by Lemma \ref{pbound-lem}
\begin{align}
\int_{2\delta\sqrt{d}}^{4\delta\sqrt{d}}  \lambda \Pr [ W_N > \lambda] d\lambda
&\leq p(N, d-1, \arccos(1/\sqrt{d})) \int_{2\delta\sqrt{d}}^{4\delta\sqrt{d}}  \lambda d \lambda\notag\\
&\leq 6 \delta^2d \left( 2 \sqrt{d} (13)^d \left(\frac{11}{12} \right)^{N/2} \right)\notag\\
&= 12 \delta^2 d^{3/2} (13)^d  \left(\frac{11}{12} \right)^{N/2}.\label{thm2-eq7}
\end{align}

\noindent {\em Step V.}  In this step we bound the integral $\int_{4\delta\sqrt{d}}^{\infty} \lambda \Pr[W_N>\lambda] d\lambda$.
For this we specialize the $\epsilon$-net argument used in the proof of Theorem \ref{main-thm1} to the case when $\{\varphi_n\}_{n=1}^N$ are
i.i.d. uniformly distributed on $\mathbb{S}^{d-1}$.

Using \eqref{cover-discrete-bnd} and \eqref{eq-useinthm2} gives
\begin{align}
\Pr [W_N>\lambda] &\leq \left( \frac{8}{\epsilon} \right)^{d-1}  \sup_{z \in \mathbb{S}^{d-1}} \left( \Pr [ z \notin T_{\epsilon}(B_1(\lambda))] \right)^N\notag\\
& \leq  \left( \frac{8}{\epsilon} \right)^{d-1} \sup_{z \in \mathbb{S}^{d-1}} \left( \Pr \left[ |\langle z, \varphi\rangle| \leq \frac{2\delta}{\lambda} + \epsilon  \right] \right)^N.
\label{thm2-eq8}
\end{align}
Since $\varphi$ is uniformly distributed on $\mathbb{S}^{d-1}$ it follows from Example \ref{unif-gencond-ex} and \eqref{C-d-estimate} that
\begin{equation} \label{thm2-eq9}
\forall z \in \mathbb{S}^{d-1}, \ \ \ \Pr \left[ | \langle z, \varphi \rangle | \leq \frac{3\delta}{\lambda} \right]   \leq \frac{3 \delta \sqrt{d}}{\lambda}.
\end{equation}
Taking $\epsilon = {\delta}/{\lambda}$ in \eqref{thm2-eq8} and using \eqref{thm2-eq9} shows that 
\begin{align}
\Pr [W_N>\lambda] &\leq  \left( \frac{8\lambda}{\delta} \right)^{d-1} \sup_{z\in\mathbb{S}^{d-1}} 
\left( \Pr \left[ | \langle z, \varphi \rangle | \leq \frac{3\delta}{\lambda}\right] \right)^N\notag\\
&\leq  \left( \frac{8\lambda}{\delta} \right)^{d-1} \left( \frac{3\delta \sqrt{d}}{\lambda} \right)^N. \label{thm2-eq10}
\end{align}
It now follows from \eqref{thm2-eq10} that if $N \geq d+2$ then
\begin{align}
\int_{4\delta\sqrt{d}}^{\infty} \lambda \Pr [ W_N > \lambda] d\lambda 
&\leq \int_{4\delta\sqrt{d}}^{\infty} \lambda  \left( \frac{8\lambda}{\delta} \right)^{d-1} \left( \frac{3\delta \sqrt{d}}{\lambda} \right)^N d\lambda \notag\\
&= \delta^2 (3\sqrt{d})^N 8^{d-1} \frac{(4\sqrt{d})^{d-N+1}}{N-d-1}\notag\\
&\leq \delta^2\sqrt{d} \left( \frac{3}{4} \right)^N \left( 32\sqrt{d} \right)^d. \label{thm2-eq11}
\end{align}

\noindent {\em Step VI.}  Combining \eqref{thm2-eq1}, \eqref{thm2-eq5b}, \eqref{thm2-eq7} and \eqref{thm2-eq11} yields

\begin{align*}
\frac{1}{2}\mathbb{E}|W_N|^2 &\leq \frac{e^{12} \delta^2 d^3}{(N+1)(N+2)}
+12 \delta^2 d^{3/2} (13)^d  \left(\frac{11}{12} \right)^{N/2}
+\delta^2\sqrt{d} \left( \frac{3}{4} \right)^N \left( 32\sqrt{d} \right)^d\\
& \leq \frac{e^{12} \delta^2 d^3}{(N+1)(N+2)} + 13\delta^2 d^{3/2}\left( \frac{11}{12} \right)^{N/2} \left( 32\sqrt{d} \right)^d\\
&=\frac{e^{12} \delta^2 d^3}{(N+1)(N+2)} + 13\delta^2 d^{3/2}\left( \frac{11}{12} \right)^{N/2} e^{\frac{d \ln(1024d)}{2}}.
\end{align*}
This completes the proof.\\
\end{proof}

\begin{corollary}
There exist absolute constants $a,c>0$ such that if $N \geq a d \ln d$
and $d \geq 2$ and
if the random vectors $\{\varphi_n\}_{n=1}^N \subset \R^d$ are i.i.d. uniformly distributed on $\mathbb{S}^{d-1}$ then 
\begin{equation}\label{cor-bnd-unif}
\mathbb{E}|W_N|^2 \leq \frac{c \thinspace d^3 \delta^2}{N^2}.
\end{equation}
\end{corollary}
\vspace{.1in}

It is instructive to compare the dimension dependence of consistent reconstruction with linear reconstruction.
\begin{example}
Fix an error tolerance $\eta>0$. Suppose that the random vectors $\{\varphi_n\}_{n=1}^N \subset \mathbb{S}^{d-1}$ are uniformly distributed on $\mathbb{S}^{d-1}$,
that $\{\epsilon_n\}_{n=1}^N$ are uniformly distributed on $[-\delta, \delta]$, and that the $\{\varphi_n\}_{n=1}^N$ and $\{\epsilon_n\}_{n=1}^N$ are mutually independent.

The bound \eqref{cor-bnd-unif} shows that if $N \geq \frac{d^{3/2}\delta \sqrt{c}}{\eta}$ then the mean squared error for a consistent estimate $\widetilde{x}_{\rm cr}$ satisfies
$$\mathbb{E}\|x - \widetilde{x}_{\rm cr} \|^2 \leq \mathbb{E}|W_N|^2 \leq \eta^2.$$  Thus
$N=\mathcal{O}(d^{3/2})$ measurements are sufficient for consistent reconstruction to achieve $\eta^2$-precise
mean squared error.

For linear reconstruction with an arbitrary dual frame, applying 
\eqref{utf-mse} with $\sigma = \delta^2/3$ shows that linearly reconstructed $\widetilde{x}_{\rm lin}$ satisfies
$$\mathbb{E}\| x - \widetilde{x}_{\rm lin}\|^2 \geq \frac{d^2\delta^2}{3N}.$$
Thus at least $N \geq \frac{d^2\delta^2}{3\eta^2}$ measurements are necessary for linear reconstruction to achieve $\eta^2$-precise mean squared error.
\end{example}

\section{Acknowledgments}

A.M. Powell was supported in part by NSF DMS 1211687 and NSF DMS 0811086, and also gratefully acknowledges the Academia Sinica Institute of Mathematics (Taipei, Taiwan) for its hospitality and support.

The authors thank Yaniv Plan, Mark Rudelson, Roman Vershynin, and Elena Yudovina for helpful comments.  The authors especially thank Elena Yudovina
for a suggestion which led to an improved proof of Theorem \ref{main-thm1} that is more precise than an earlier version and that is also 
more consistent with the proof of Theorem \ref{mainthm2}.

\end{document}